\documentclass[11pt, a4paper,onecolumn]{belarticle4}
\pdfoutput=1
\usepackage{revsymb, amsmath, amsfonts, amssymb, enumerate, fullpage, amsthm, graphicx, braket, relsize, mathrsfs}
\usepackage{graphicx,color}
\usepackage{paralist}
\usepackage[export]{adjustbox}
\usepackage[linktocpage=true, colorlinks=true, linkcolor=blue, urlcolor=blue, citecolor=blue]{hyperref}
\usepackage[numbers,sort&compress]{natbib} %% so that it cites e.g. 15-32 rather than listing all.

\usepackage[svgnames]{xcolor}
\usepackage{tikz}
\usetikzlibrary{decorations.markings}
\usetikzlibrary{shapes.geometric}
\pgfdeclarelayer{edgelayer}
\pgfdeclarelayer{nodelayer}
\pgfsetlayers{edgelayer,nodelayer,main}
\tikzstyle{none}=[inner sep=0pt]
\tikzstyle{simple}=[-,draw=Black,line width=2.000]
\newtheorem{theo}{Theorem}%[subsection]
\newtheorem{thm}[theo]{Theorem}
\newtheorem{prop}[theo]{Proposition}

\newtheorem{defn}[theo]{Definition}

\newcommand{\cH}{\mathcal{H}}

\newcommand{\id}{\mathbb{I}}
\newcommand{\tr}[2]{\mathrm{tr}_{#2} \left\{ #1 \right\}}
\newcommand{\trace}[1]{\mathrm{tr}\left\{#1 \right\}}

\newcommand{\cQ}{\mathcal{Q}}
\newcommand{\cP}{\mathcal{P}}

\newcommand{\A}{\mathbb{A}}
\newcommand{\X}{\mathbb{X}}
\newcommand{\Ass}[2]{\boldsymbol{\Sigma}_{#1|#2}}
\newcommand{\Np}{\mathrm{N}}

\newcommand{\Y}{\mathbb{Y}}
\newcommand{\Bo}{\mathbb{B}}
\newcommand{\pr}{\text{\textbf{P}}}

\definecolor{orangy}{RGB}{213,94,0}

\newcommand{\blk}{\color{black}}

\begin{document}
	
	\title{On characterising assemblages in Einstein-Podolsky-Rosen scenarios}
	
	\author{Vinicius P.~Rossi}
	\affiliation{International Centre for Theory of Quantum Technologies, University of  Gda{\'n}sk, 80-308 Gda{\'n}sk, Poland}
	\email{vinicius.prettirossi@phdstud.ug.edu.pl}
	\author{Matty J.~Hoban}
	\affiliation{Cambridge Quantum Computing Ltd}
	\affiliation{Quantinuum LLC}
	\affiliation{Department of Computing, Goldsmiths, University of London, New Cross, London SE14 6NW, United Kingdom}
	\author{Ana Bel\'en Sainz}
	\affiliation{International Centre for Theory of Quantum Technologies, University of  Gda{\'n}sk, 80-308 Gda{\'n}sk, Poland}

	\date{\today}

	\begin{abstract}
		Characterising non-classical quantum phenomena is crucial not only from a fundamental perspective, but also to better understand its capabilities for information processing and communication tasks. In this work, we focus on exploring the characterisation of Einstein-Podolsky-Rosen inference (a.k.a.~steering): a signature of non-classicality manifested when one or more parties in a Bell scenario have their systems and measurements described by quantum theory, rather than being treated as black boxes. We propose a way of characterising common-cause assemblages from the correlations that arise when the trusted party performs tomographically-complete measurements on their share of the experiment, and discuss the advantages and challenges of this approach. Within this framework, we show that so-called almost quantum assemblages satisfy the principle of macroscopic noncontextuality, and demonstrate that a subset of almost quantum correlations recover almost quantum assemblages in this approach.
	\end{abstract}
	
	\maketitle
	
	\tableofcontents
	
	\newpage
	
	\section{Introduction}\label{intro}
	
	\quad Identifying the non-classical aspects of quantum theory is not only an important research theme in the foundations of quantum theory, but it is now of relevance to quantum computation and quantum information. In particular, these non-classical aspects can be seen as a resource that can be used for a quantum advantage in computation \cite{howard,raussendorf,andersbrowne,hobanbrowne} or secure communication \cite{branciard2012one,qkd,vaziranividick}. Of special interest is the study of non-classical resources in Bell scenarios, often termed ``Bell non-locality", which is a resource in protocols for random number generation \cite{pironio, coudron, colbeck}, quantum key distribution \cite{qkd,vaziranividick}, and verifiable, delegated quantum computation \cite{ruv}. Interestingly, the protocols where this non-classicality is a resource are \textit{device-independent}, which means that (possibly) quantum devices can be treated as multi-party \textit{black boxes} that have classical inputs and outputs: one does not need to specify \textit{a priori} the quantum description of the devices inside the black box. In the nomenclature of Ref.~\cite{cowpie}, these resources can be termed \textit{common-cause boxes}, which comes from the structure of a Bell experiment \cite{Bell64}. 
	
	Importantly, when considering resources, one must understand their limitations as well as their possibilities for enhancing technological performance. The characterisation of the set of common-cause boxes possible in quantum theory is useful for understanding the scope of device-independent information processing since it can, for instance, limit the ability of an eavesdropper to predict the output of a black box and thus lead to private randomness generation \cite{pironio}. One important example where we see the limitations of quantum systems is in the Tsirelson bound \cite{cirel1980quantum}, which dictates the largest violation possible with quantum common-cause boxes of the Clauser-Horne-Shimony-Holt (CHSH) inequality \cite{chsh}. The impact of this limitation includes that Alice and Bob cannot use quantum common-cause boxes to enhance communication so to transmit arbitrary messages with only one bit of classical communication \cite{van2013implausible}. In turn, quantum common-cause boxes that give this maximal violation can be used for all of the kinds of the aforementioned device-independent protocols. 
	
	Curiously, the Tsirelson bound is not the largest numerical violation possible, and such \textit{post-quantum} violations can be achieved with common-cause boxes that respect the no-signalling principle: inputs for one party do not influence the statistics of another party's output. Popescu and Rohrlich used this example (now called the PR box) to point out that for these boxes, special relativity is not enough to constrain one to the set of quantum boxes, and asked what axioms or principles could single out this set \cite{popescu1994quantum}. Answering such a question would give us deep insight into why quantum theory is one of our ultimate theories, and explain its limitations for information processing.
	
	Subsequent to the work of Popescu and Rohrlich, van Dam described an \textit{information theoretic} principle that ruled out the PR box: the principle of Non-Trivial Communication Complexity \cite{van2013implausible}. Despite follow-up work by Brassard \textit{et al} \cite{Brassard06} it was unclear if this principle could recover quantum common-cause boxes. Significant progress on this question was made by the seminal work of Navascu\'{e}s, Pironio and Ac\'{i}n (NPA) \cite{NPA07,NPA09,PNA10}, whom developed a hierarchy of semi-definite programs that converge to a set of quantum common-cause boxes\footnote{Technically, the hierarchy converges to the set of common-cause boxes realised by commuting measurements on a quantum state, which, for infinite dimensional systems, is distinct from boxes produced by local measurements defined according to a tensor product. For finite dimensional quantum systems, the two definitions are equivalent.}.
	
	The NPA hierarchy is an extremely useful tool, but is not described in terms of either information-theoretic or physical principles. As a result, many principles followed such as Information Causality \cite{IC}, Macroscopic Locality \cite{ML}, Local Orthogonality \cite{LO,LO2}, and No Common Certainty of Disagreement \cite{contreras2021observers}. An obstacle to the development of such principles is the almost-quantum set of common-cause boxes \cite{AQ}, which have not yet violated any of the proposed principles.
	
	Independently of the study of common-cause boxes, another line of research in this topic pertains to the characterisation of not only common-cause boxes, but of the whole of quantum theory.
	Unlike for common-case boxes, various axiomatisations of (finite-dimensional) quantum theory have been successfully developed \cite{hardy2001quantum,clifton2003characterizing,goyal2008information,chiribella2020probabilistic,hardy2011reformulating,masanes2011derivation,chiribella2011informational,dakic2011quantum,hardy2012operator,hardy2013formalism,budiyono2017quantum,hohn2017quantum,Hohn2017toolbox,Selby2021reconstructing,wilce2018royal,Wetering2019effecttheoretic,nakahira2020derivation,tull2020categorical}, especially in the study of \textit{generalised probabilistic theories} \cite{barrett2007gpt} or \textit{process theories} \cite{coecke_kissinger_2017}. Furthermore, tentative progress has been made in understanding the information processing power of quantum theory within such a broad framework \cite{barrett2007gpt, barrett2019computational}. The starting point for these frameworks requires one, in some sense, to dictate the degrees of freedom in a general, possibly non-quantum, experiment from the outset so it can be associated with some vector space. This approach is thus very different from the device-independent approach as it requires a characterisation of the systems involved. As a result, this approach has not led to deep intuitions about the strengths or limitations of quantum phenomena such as non-classical common-cause boxes.
	
	These two approaches to characterising quantum systems within a more general framework -- a.k.a.~characterising quantum theory ``from the outside" -- are not the only possibilities. In particular, one situation is to start with some elements of quantum theory, and then beyond this to allow something more general. A notable example in this direction is the development of \textit{process matrices} that exhibit indefinite causal order \cite{oreshkov}. Other examples include the study of Bell scenarios assuming only local measurements \cite{barnum,unified,geller2014quantifying} and recovering the measurement postulates of quantum theory from just unitary dynamics \cite{masanes}. In this work, we also pursue a direction of assuming some aspect of quantum theory, in particular we assume that quantum theory holds locally for a single party alone. 
	
	To be more precise, instead of considering Bell scenarios, we consider Einstein-Podolsky-Rosen (EPR) scenarios \cite{schrodinger1936probability,wiseman2007steering}, inspired by their seminal 1935 paper \cite{einstein1935can}. In such a scenario, as mentioned, there is a single party (called Bob) that has a quantum system with known degrees of freedom (with a particular Hilbert space dimension) that can be directly measured; thus the system has a quantum-theoretic description. Outside of Bob's laboratory we can treat everything else as uncharacterised, but satisfying causal constraints. In particular, there are multiple, non-communicating parties that generate classical data. In this way we can see the EPR scenario as a Bell scenario with these extra assumptions made about one particular party. While in Bell scenarios one considers boxes associated with certain correlations, in EPR scenarios the fundamental object is an \textit{assemblage} \cite{pusey2013negativity}. An assemblage is a collection of (sub-normalised) conditional density matrices: a description of the quantum system in Bob's laboratory conditioned on the classical data generated outside of this laboratory. An assemblage provides a rich structure since it can capture both box-like correlations as well as quantum information. 
	
	Non-classical\footnote{In the literature, these assemblages are also referred to as `steerable' assemblages.} assemblages in Einstein-Podolsky-Rosen scenarios\footnote{In the literature, these are also known as `steering' scenarios.} show a particular non-classical aspect of nature that quantum theory features in the lab. In addition, they have been shown useful for enhancing our performance at cryptographic protocols \cite{branciard2012one,law2014quantum,passaro2015optimal,cavalcanti2016quantum}, hence manifesting as useful quantum resources. To capture this non-classical resource one needs to characterise the set of classical assemblages, and much work has been devoted to this \cite{weight,cavalcanti2015detection, cavalcanti2016quantum}. However, the question of `how to characterise non-classical assemblages from the outside' has barely been explored. That is, given a general framework for non-classical assemblages can we single out those with a completely quantum explanation from those that go beyond the quantum formalism? As well as being fundamental, this question gives us another avenue to explore the resourcefulness of non-classical assemblages. In this work we comment on the challenges we face when tackling this question, on a particular way to overcome these, and on the new insights from and limitations of our approach.
	
	\section{Einstein-Podolsky-Rosen scenarios: preliminaries}
	
	\quad Einstein-Podolsky-Rosen (EPR) scenarios correspond to a particular type of experiment where non-classical features of nature manifest, and captures a particular type of non-classical property referred to as \textit{EPR inference}\footnote{In the literature, EPR inference is also known as `EPR steering' or merely `steering'.}. 
	Some features in the specification of an EPR scenario are similar to those which specify Bell scenarios, so the reader familiar with the latter might find a similar discourse below. 
	
	For simplicity, let us first introduce the case of two parties, hereon called Alice and Bob (see Fig.~\ref{fig:EPRbip}). Alice and Bob share a physical system prepared on a state possibly unknown to them, and Alice wants to learn about the system held by Bob in his lab. Alice and Bob are distant from each other and cannot communicate (neither classically nor with quantum systems), hence the only way that Alice can learn about the state of Bob's system is by performing measurements on her share of the system and making inferences about Bob's. After performing a measurement, then, Alice can update her knowledge on the state preparation of Bob's quantum system conditioned on her obtained outcome. In this sense, each measurement choice by Alice selects an ensemble of the possible updated states of Bob's quantum system, and the probability of each of her outcomes yields the probability with which each of those updated states will arise. The collection of these ensembles is called \textit{assemblage} and is the object of interest in the study of EPR scenarios.
	Notice that this discourse is different from the ubiquitous one where EPR scenarios are about non-classical causal influences, the famous `spooky action at a distance'. For a deeper discussion on the advantages of taking this less-common inferential perspective we refer the reader to Ref.~\cite{zjawin2021quantifying}.
	
	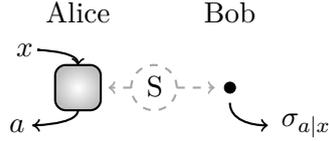
\begin{figure}
		\begin{center}
			\begin{tikzpicture}[scale=1]
				\node (A) at (-1,1) {Alice};
				\shade[draw, thick, ,rounded corners, inner color=white,outer color=gray!50!white] (-0.7,-0.3) rectangle (-1.3,0.3) ;
				\draw[thick, ->] (-1.5,0.5) to [out=180, in=90] (-1,0.3);
				\node at (-1.7,0.5) {$x$};
				\draw[thick, -<] (-1,-0.3) to [out=-90, in=180] (-1.5,-0.5);
				\node at (-1.8,-0.5) {$a$};
				
				\node (B) at (1,1) {Bob};
				\draw[thick, ->] (1,-0.2) to [out=-90, in=180] (1.5,-0.5);
				\node at (2,-0.5) {$\sigma_{a|x}$};
				\node[draw,shape=circle,fill,scale=.4] at (1,0) {};
				
				\node at (0.01,0.02) {S};
				\draw[thick, dashed, color=gray!70!white] (0,0) circle (0.3cm);
				\draw[thick, dashed, color=gray!70!white, ->] (-0.3,0) -- (-0.6,0);
				\draw[thick, dashed, color=gray!70!white, ->] (0.3,0) -- (0.8,0);
			\end{tikzpicture}
		\end{center}
		\caption{\textbf{Bipartite EPR scenario:} Alice refines her knowledge about the state of Bob's quantum system by making inferences from performing measurements on her share of the system. Alice's measurement choices are labelled by the classical variable $x$ that takes values in the set $\X$, and the measurement outcomes are labelled by the classical variable $a$ that takes values in the set $\A$. The possibly subnormalised positive-semidefinite matrix $\sigma_{a|x}$ represents the state of Bob's system when Alice's measurement choice and outcome are $x$ and $a$ respectively. The probability that Alice obtains $a$ when performing measurement $x$ is given by $p(a|x)=\tr{\sigma_{a|x}}{}$. The assemblage (ensemble of ensembles) representing the information on the quantum state of Bob's system is $\Ass{\A}{\X}=\{\{ \sigma_{a|x}\}_{a\in\A}\}_{x\in\X}$.}
		\label{fig:EPRbip}
	\end{figure}
	
	\subsection{Bipartite EPR scenarios}
	
	\quad In a bipartite scenario, we denote by $\X$ the set of classical labels for Alice's measurement choices, and $\A$ the set of classical labels for her measurement outcomes\footnote{In principle, different measurement can have different number of outcomes, but for the current discussion we can take these sets to be the same (and all equal to $\A$) without loss of generality.}. In addition, let us denote by $\cH_B$ the Hilbert space corresponding to Bob's quantum system, which is moreover known to both parties. 
	An assemblage, denoted by $\Ass{\A}{\X}$, is hence given by the collection of (possibly unnormalised) quantum states $\Ass{\A}{\X}=\{\{ \sigma_{a|x}\}_{a\in\A}\}_{x\in\X}$, where $p(a|x)=\tr{\sigma_{a|x}}{}$ is the probability that Alice obtains outcome $a \in \A$ when performing measurement $x \in \X$, and $\rho_{a|x}=\frac{\sigma_{a|x}}{p(a|x)}$ is the normalised quantum state that refines Alice's knowledge on the state of Bob's quantum system. 
	
	An assemblage $\Ass{\A}{\X}$ is said to admit a \textit{quantum realisation} if there exists a Hilbert space $\cH_A$ for Alice, a set of measurements $\{\{M_{a|x}\}_{a\in\A}\}_{x\in\X}$ in $\cH_A$, and a quantum state $\rho$ in $\cH_A \otimes \cH_B$, such that the elements of $\Ass{\A}{\X}$ can be expressed as: 
	\begin{align}
		\sigma_{a|x} = \tr{M_{a|x} \otimes \id_{\cH_B} \, \rho}{A}\quad \forall \, a\in\A\,,\, x\in\X\,.
	\end{align}
	
	\subsection{Assemblages without a quantum realisation}
	
	\quad In this work we are interested in characterising the set of quantum-realisable assemblages by singling them out within a broader class of assemblages that may include non-quantum-realisable ones (known as post-quantum assemblages). Two comments are in order for understanding this question: how may a post-quantum assemblage arise within such a seemingly-quantum-exclusive experiment?, and how are post-quantum assemblages mathematically specified? 
	
	Regarding the first question, there are different ways in which one can capture what is necessarily quantum in this experiment, and what aspects of it may be dictated by a beyond-quantum theory. Informally, one may think of a universe that can be fully fleshed out locally using a quantum description, but that holistically requires something beyond the quantum formalism to be described and understood. Alternatively, one can also think of the constraints of an EPR scenario as just one more way to probe nature, and explore the scope of lessons to learn from it. This latter approach takes on the device-independent approach underpinning a Bell scenario, where even if the parties had access to a quantum system and measurements, the only information that they rely on to assess the properties of nature are the classical labels of the measurement choices and outcomes, and the statistics one may draw from them (no need to assume or use the quantum formalism at all!). So, given that an EPR scenario is defined by $\X$, $\A$, and $\cH_B$, then considerations beyond quantum theory may come in in the nature of the system shared by Alice and Bob. This perspective goes along the lines of a causal grounding of an EPR scenario, and gauges the classicality or quantumness of an assemblage in terms of the common causes (which may be thought of as the shared system) that correlate Alice and Bob's labs. Another possibility comes from the constraint that in an EPR scenario Alice and Bob are distant parties that perform space-like separated local actions and cannot communicate with each other. This also opens the door for some non-quantumness to sneak in in the form of fine-tuned hidden signalling. However, given the conceptual problems that fine-tuning comes with \cite{wood2015lesson} we encourage the reader not to entertain this latter approach. Nevertheless, a perspective provided in terms of imposing an overall no-signalling principle between Alice and Bob allows one to include in the picture assemblages with no quantum realisation \cite{sainz2015postquantum,hoban2018channel,sainz2018formalism,sainz2020bipartite}, in the same way that in Bell scenarios post-quantum correlations (such as the one given by Popescu and Rohrlich \cite{popescu1994quantum}) are encompassed. 
	
	Regardless of which fundamental way one interprets the possibility of there being EPR inferences that have no quantum explanation, there is the formal question of how are general assemblages (be it quantum or post-quantum) mathematically specified. Luckily, the two ways described above (common-cause perspective and no-signalling perspective) define the same set of mathematical objects as their most general set of assemblages \cite{pauloGPT2}. Hence, one can readily explore the consequences of post-quantum EPR inferences while mulling over its philosophical implications. In the case of a bipartite EPR scenario, the definition of a general assemblage is as follows: 
	
	\begin{defn}\label{def:ccbip}\textbf{Common-cause assemblage, a.k.a.~non-signalling assemblage.--}\\
		Consider a bipartite EPR scenario, where Alice's measurements and outcomes are labelled by the elements of the sets $\X$ and $\A$, respectively, and Bob's quantum system is represented by the Hilbert space $\cH_B$. Then, an assemblage $\Ass{\A}{\X}$ is a common-cause assemblage (equivalently, a non-signalling assemblage) if the following constraints are satisfied: 
		\begin{align}\label{eq:psdns}
			\sigma_{a|x} &\geq 0 \quad \forall \, a\in\A\,,\, x\in\X\,,\\ \label{eq:NScons}
			\sum_{a} \sigma_{a|x} &= \sum_{a} \sigma_{a|x'} \quad \forall \, a\in\A\,,\, x,x' \in\X\,,\\
			\tr{\sum_{a} \sigma_{a|x}}{} &= 1 \quad \forall \, x \in \X \,. \label{eq:normNS}
		\end{align}
	\end{defn}
	
	A celebrated theorem by Gisin \cite{gisin1989stochastic} and Hughston, Jozsa, and Wootters \cite{hughston1993complete} (GHJW) shows that post-quantum assemblages in the traditional bipartite setting specified above cannot occur, that is, any common-cause assemblage given by Def.~\ref{def:ccbip} admits a quantum realisation (See Theorem \ref{GHJW}). However, recent research has shown that in multipartite EPR scenarios \cite{sainz2015postquantum,sainz2018formalism,hoban2018channel} or modified bipartite ones \cite{sainz2020bipartite} post-quantum steering may arise. In this work we focus on the former case, which we describe below. 
	
	\subsection{Multipartite EPR scenarios}
	
	\quad We focus on multipartite EPR scenarios that consist of multiple Alice-type parties and one Bob. More specifically, consider the case where we have $\Np+1$ parties: one of them has a system represented by a Hilbert space $\cH_B$ (the party called Bob), and each of the remaining parties has two associated sets of classical variables, $\X_k$ and $\A_k$ (with $k \in \{1,\ldots,\Np \}$ denoting the party), which represent the classical labels of their measurement choices and outcomes. In analogy with the bipartite case, we regard this parties as `the Alices', and refer to a particular one by Alice$_{k}$. All parties are distant, perform space-like separated actions on their share of a physical system, and cannot communicate to each other. By performing measurements, then, the Alices update their knowledge on the state preparation of Bob's system. An assemblage $\Ass{\vec{\A}}{\vec{\X}}$ in this scenario is then given by the collection of (possibly unnormalised) quantum states $\{\sigma_{a_1\ldots a_\Np | x_1 \ldots x_\Np}\}_{a_k \in \A_k\,,\, x_k \in \X_k\,,\, k \in \{1,\ldots,\Np \}}$, and $\vec{\A}$ and $\vec{\X}$ are short notation for the arrays $(\A_1,\ldots,\A_\Np)$ and $(\X_1,\ldots,\X_\Np)$, respectively. 
	
	An assemblage $\Ass{\vec{\A}}{\vec{\X}}$ is said to admit a \textit{quantum realisation} if there exists a Hilbert space $\cH_{A_k}$ for each Alice, a set of measurements $\{\{M^{(k)}_{a|x}\}_{a\in\A_k}\}_{x\in\X_k}$ in $\cH_{A_k}$ for each $k \in \{1,\ldots,\Np \}$, and a quantum state $\rho$ in $\otimes_{k=1}^{\Np} \cH_{A_k} \otimes \cH_B$, such that the elements of $\Ass{\vec{\A}}{\vec{\X}}$ can be expressed as: 
	\begin{align}
		\sigma_{a_1\ldots a_\Np | x_1 \ldots x_\Np} = \tr{\otimes_{k=1}^{\Np} M^{(k)}_{a_k|x_k} \otimes \id_{\cH_B} \, \rho}{A_1\ldots A_{\Np}}\quad \forall \, a_k \in \A_k\,,\, x_k \in \X_k\,.
	\end{align}
	
	In these scenarios, then, the most general assemblage that one can mathematically write while complying with the operational constraints of the scenario are given by the following. 
	
	\begin{defn}\label{def:ccmp}\textbf{Common-cause assemblage, a.k.a.~non-signalling assemblage -- multipartite EPR scenarios.--}\\
		Consider a multipartite EPR scenario, where the Alices' measurements and outcomes are labelled by the elements of the sets $\X_k$ and $\A_k$, respectively, for each $k \in \{1,\ldots,\Np \}$,  and Bob's quantum system is represented by the Hilbert space $\cH_B$. Then, an assemblage $\Ass{\vec{\A}}{\vec{\X}}$ is a common-cause assemblage (equivalently, a non-signalling assemblage) if the following constraints are satisfied: 
		\begin{align}
			\sigma_{a_1\ldots a_\Np | x_1 \ldots x_\Np} &\geq 0 \quad \forall \, a_k \in \A_k\,,\, x_k \in \X_k\,,\\
			\sum_{a_k \in \A_k} \sigma_{a_1 \ldots a_k \ldots a_\Np|x_1 \ldots x_k \ldots x_\Np} &=  \sum_{a_k \in \A_k} \sigma_{a_1 \ldots a_k \ldots a_\Np|x_1 \ldots x'_k \ldots x_\Np} \\ \nonumber &  \forall \, a_j\in\A_j \, \text{and} \, x_{j} \in\X_j \, \text{with} \, j \neq k\,,\, x_k,x'_k \in \X_k\,,\, k \in \{1,\ldots,\Np \} \,,\\
			\tr{\sum_{a_1 \ldots a_k \ldots a_\Np} \sigma_{a_1 \ldots a_k \ldots a_\Np|x_1 \ldots x_k \ldots x_\Np}}{} &= 1 \quad \forall \, x_k \in \X_k \,.
		\end{align}
	\end{defn}
	
	\bigskip
	
	Quantum-realisable assemblages and common-cause assemblages are not the only sets of assemblages of interest in the literature and of relevance in this work. One particular set of assemblages that is relevant is that of \textit{almost-quantum} assemblages, which is a set slightly larger than that of quantum-realisable assemblages. The relevance of almost-quantum assemblages is that they are conveniently defined such that testing whether an assemblage admits an almost-quantum realisation amounts to a single instance of a semidefinite program, hence it serves as a convenient numerical tool to outer bound quantum-realisable assemblages and, from an information-theoretical perspective, upper-bound their resourcefulness \cite{cavalcanti2016quantum}. In the next subsection, we introduce the almost-quantum assemblages, and the moment matrix formalisation to which they are equivalent, as this makes the link to semidefinite programming.

	\subsection{Almost-quantum assemblages and moment matrices}\label{ap:aqstuff}
	
	The set of almost-quantum assemblages is strongly connected to the set of almost-quantum correlations \cite{AQ}. Because of this close connection, we will discuss both in this section. As shown in Ref.~\cite{hoban2018channel}, this set of assemblages has two equivalent definitions, one of which is defined in terms of moment matrices. Here we opt to introduce their definition in terms of moment matrices, which is linked to semi-definite programming, since this proves convenient for this work. To begin, we will need some useful notation associated with EPR (and Bell) experiments.
	
	\begin{defn}\textbf{Alphabet, properties of words, and equivalence relations between words.-- }\\
		We begin by defining an alphabet, some of the properties of the words you can write with it, and equivalence relations between words. Consider the case where we have $\Np$ black-box parties in our experiment, which for consistency we refer to as Alices. Let $\X$ and $\A$ be the sets of measurement choices for each Alice, and measurement outcomes for each measurement, respectively\footnote{As before, for the purpose of this work and simplicity in the presentation, it is convenient to take the set of measurements to be the same for all Alices, and the set of outcomes to be the same for all measurements.}. The alphabet $\Upsilon_\mathrm{E}$, whose letters are of the form $a|x$, is defined as:   
		\begin{align}
			\Upsilon_\Np := \bigcup_{k =1:\Np} \{ a_k|x_k \}_{a_k \in \A\,,\, x_k \in \X}\,.
		\end{align}
		A \textbf{word} is a concatenation of elements of $\Upsilon_\Np$. Given two words $\mathbf{v},\mathbf{w}$, their concatenation $\mathbf{vw}$ yields another word. The word $\mathbf{v}^\dagger$ denotes the word given by the letters of $\mathbf{v}$ written in reverse order. The symbol $\emptyset$ denotes the \textbf{empty word}, and has length $0$. Finally, the set $\mathcal{S}^*_\mathrm{E}$ is the \textbf{set of words} of arbitrary length you may write from the letter of $\Upsilon_\Np$ and the empty word $\emptyset$. 
		
		Equivalence relations (a.k.a.~\textbf{symmetry operations}) among words are the following:
		\begin{compactitem}
			\item $\mathbf{vw} = \mathbf{v \emptyset w}$ $\forall$ $\mathbf{v},\mathbf{w} \in \mathcal{S}^*_\mathrm{E}$. 
			\item $\mathbf{vv} = \mathbf{v}$ $\forall$ $\mathbf{v} \in \Upsilon_\Np$. 
			\item $a_k|x_k \, a_{k'}|x_{k'} = a_{k'}|x_{k'} \, a_k|x_k $ $\forall$ $k \neq k' \in \{1,\ldots,\Np\}$.
		\end{compactitem}
		
		A word $\mathbf{v}$ is \textbf{null} if, after applying symmetry operations, there is a letter $a_k|x_k$ followed by a letter $a^\prime_k|x_k$ with $a_k \neq a^\prime_k$. 
		
	\end{defn}
	
	\begin{defn}\textbf{Almost-quantum set of words.--}\label{def:AQw} \\
		Let $S \subseteq \{1, \ldots, \Np\}$ be a subset of the black-box parties. Let $\mathbf{a}_S|\mathbf{x}_S$ denote the word given by $\mathbf{a}_S|\mathbf{x}_S = a_{i_1}|x_{i_1} \, \ldots \, a_{i_{|S|}}|x_{i_{|S|}}$, where $a_k|x_k \in \Upsilon_\Np \quad \forall \, k \in S$.  
		
		The set of words $\mathcal{S}_{\mathrm{AQ}}$ is given by
		\begin{align}
			\mathcal{S}_{\mathrm{AQ}} := \{ \emptyset \} \cup \left\{ \mathbf{a}_S|\mathbf{x}_S \quad \Big\lvert \quad  a_k|x_k \in \Upsilon_\Np \quad \forall \, k \in S\,,\quad   S \subseteq \{1, \ldots, \Np\} \right\}\,.
		\end{align}
	\end{defn}
	
	The following definitions now are particular to correlations in a Bell experiment, or assemblages. When talking about Bell scenarios, we will consider an $\Np$-partite scenario with $\Np$ black-box parties, which for consistency we will call Alices. A correlation in such multipartite Bell scenario will be denoted by $\pr_{\vec{\A}|\vec{\X}}$, where $\vec{\A}$ is an $\Np$ component vector, where each entry is the set $\A$, and similarly for $\vec{\X}$. For an EPR scenario, we will focus on one with $\Np$ black-box parties (Alices) and one Bob (whose quantum system has finite dimension $d$).  An assemblage in such EPR scenario will be denoted by $\Ass{\vec{\A}}{\vec{\X}}$.
	
	\begin{defn}\textbf{Bell-scenario almost-quantum moment matrix .--} \\ 
		Let $\Gamma_{\mathrm{AQ}}^{\mathrm{B}}$ be a square matrix of size  $|\mathcal{S}_{\mathrm{AQ}}| \times |\mathcal{S}_{\mathrm{AQ}}|$, whose entries are complex numbers. Let the rows and columns of $\Gamma_{\mathrm{AQ}}^{\mathrm{B}}$ be labelled by the words in $\mathcal{S}_{\mathrm{AQ}}$.  $\Gamma_{\mathrm{AQ}}^{\mathrm{B}}$ is an almost-quantum moment matrix in the Bell scenario iff it satisfies the following properties: 
		\begin{align}
			&\Gamma_{\mathrm{AQ}}^{\mathrm{B}} \geq 0 \,, \\
			&\Gamma_{\mathrm{AQ}}^{\mathrm{B}} (\emptyset,\emptyset) = 1 \,,\\
			&\Gamma_{\mathrm{AQ}}^{\mathrm{B}} (\mathbf{v},\mathbf{w}) = \Gamma_{\mathrm{AQ}}^{\mathrm{B}} (\mathbf{v'},\mathbf{w'}) \quad \text{if} \quad \mathbf{v}^\dagger\mathbf{w} = \mathbf{v'}^\dagger\mathbf{w'} \,,\\
			&\Gamma_{\mathrm{AQ}}^{\mathrm{B}} (\mathbf{v},\mathbf{w}) = 0 \quad \text{if} \quad \mathbf{v}^\dagger\mathbf{w} \quad \text{is null} \,.
		\end{align}
	\end{defn}
	
	\begin{defn}\textbf{Almost-quantum correlation.--}\\
		A correlation $\pr_{\vec{\A}|\vec{\X}}$ is almost-quantum if there exists an almost-quantum moment matrix $\Gamma_{\mathrm{AQ}}^{\mathrm{B}}$ such that: 
		\begin{align}
			\Gamma_{\mathrm{AQ}}^{\mathrm{B}}(\emptyset,\mathbf{a}_S|\mathbf{x}_S) = p_S(\mathbf{a}_S|\mathbf{x}_S)\,, \quad \forall \mathbf{a}_S|\mathbf{x}_S \in \mathcal{S}_{\mathrm{AQ}}\,,
		\end{align}
		where $p_S(\mathbf{a}_S|\mathbf{x}_S)$ is the marginal conditional probability distribution of $\pr_{\vec{\A}|\vec{\X}}$ for the parties in the set $S$.
	\end{defn}
	
	\begin{defn}\textbf{EPR-scenario almost-quantum moment matrix .--}\label{def:aqassemmm} \\ 
		Let $\Gamma_{\mathrm{AQ}}^{\mathrm{EPR}}$ be a square matrix of size  $|\mathcal{S}_{\mathrm{AQ}}| \times |\mathcal{S}_{\mathrm{AQ}}|$, whose entries are $d \times d$ complex matrices. Let the rows and columns of $\Gamma_{\mathrm{AQ}}^{\mathrm{EPR}}$ be labelled by the words in $\mathcal{S}_{\mathrm{AQ}}$.  $\Gamma_{\mathrm{AQ}}^{\mathrm{EPR}}$ is an almost-quantum moment matrix in the EPR scenario iff it satisfies the following properties: 
		\begin{align} \label{eq:28}
			&\Gamma_{\mathrm{AQ}}^{\mathrm{EPR}} \geq 0 \,, \\ \label{eq:29}
			&\Gamma_{\mathrm{AQ}}^{\mathrm{EPR}} (\mathbf{v},\mathbf{w}) = \Gamma_{\mathrm{AQ}}^{\mathrm{EPR}} (\mathbf{v'},\mathbf{w'}) \quad \text{if} \quad \mathbf{v}^\dagger\mathbf{w} = \mathbf{v'}^\dagger\mathbf{w'} \,,\\ \label{eq:thetrace1}
			&\tr{\Gamma_{\mathrm{AQ}}^{\mathrm{EPR}} (\emptyset,\emptyset)}{} = 1 \\ \label{eq:30}
			&\Gamma_{\mathrm{AQ}}^{\mathrm{EPR}} (\mathbf{v},\mathbf{w}) = \mathbf{0} \quad \text{if} \quad \mathbf{v}^\dagger\mathbf{w} \quad \text{is null} \,,
		\end{align}
		where $\mathbf{0}$ is a $d \times d$ matrix whose entries are all $0$. 
	\end{defn}
	
	\begin{defn}\textbf{Almost-quantum assemblage.--}\\
		An assemblage $\Ass{\vec{\A}}{\vec{\X}}$ is almost-quantum if there exists an almost-quantum moment matrix $\Gamma_{\mathrm{AQ}}^{\mathrm{EPR}}$ such that: 
		\begin{align}\label{eq:31}
			\Gamma_{\mathrm{AQ}}^{\mathrm{EPR}}(\emptyset,\emptyset) = \rho_R \,,
		\end{align}
		where $\rho_R = \sum_{a_1, \ldots\, a_{\Np} \in \A} \sigma_{a_1,\ldots,a_\Np|x_1,\ldots,x_\Np}$ is the reduced state of Bob's system, and
		\begin{align}\label{eq:32}
			\Gamma_{\mathrm{AQ}}^{\mathrm{EPR}}(\emptyset,\mathbf{a}_S|\mathbf{x}_S) = \sigma_S(\mathbf{a}_S|\mathbf{x}_S)\,, \quad \forall \mathbf{a}_S|\mathbf{x}_S \in \mathcal{S}_{\mathrm{AQ}}\,,
		\end{align}
		where $\sigma_S(\mathbf{a}_S|\mathbf{x}_S)$ is the marginal assemblage element of $\Ass{\vec{\A}}{\vec{\X}}$ for the parties in the set $S$:
		\begin{align}
			\sigma_S(\mathbf{a}_S|\mathbf{x}_S) = \sum_{a_k \in \A \,:\, k \not\in S} \sigma_{a_1,\ldots,a_\Np|x_1,\ldots,x_\Np}\,.
		\end{align}
	\end{defn}
	Notice that Bob's reduce state $\rho_R$ as well as the marginal assemblage elements $\sigma_S(\mathbf{a}_S|\mathbf{x}_S)$ are well defined since we are working with assemblages $\Ass{\vec{\A}}{\vec{\X}}$ that are non-signalling. As a consequence, we have that all almost-quantum assemblages are non-signalling assemblages. However, in general scenarios the converse is not true \cite{sainz2015postquantum}. It can also be shown that all quantum assemblages are almost-quantum assemblages \cite{sainz2015postquantum,hoban2018channel}. Thus, this set is very useful for understanding how to recover the set of a quantum assemblages from a more general set.
	
	\section{Characterising assemblages}
	
	\quad Let us begin this section by discussing the question of characterising quantum correlations in Bell scenarios. Notice the focus on `the question' rather than `the attempted answers'. For simplicity and concreteness, let us base the discussion on the case of correlations in a bipartite Bell scenario. In this case, we have two distant parties -- Alice and Bob -- which share a physical system and perform measurements on it. The object of interest here is the correlations between their measurement outcomes, captured by the conditional probability distribution $\pr_{\A\Bo|\X\Y} = \{\{p(ab|xy)\}_{a\in\A\,,\, b \in \Bo}\}_{x \in \X\,,\, y \in \Y}$, where $\X$ (resp.~$\Y$) is the set of classical labels of Alice's (resp.~Bob's) measurement choices, and $\A$ (resp.~$\Bo$) denotes the set of classical labels of Alice's (resp.~Bob's) measurement outcomes\footnote{In principle, different measurement can have different number of outcomes, but for the current discussion we can take these sets to be the same without loss of generality.}. Notice that in this description of a Bell scenario, the party named Bob plays a similar role to that of Alice, unlike in an EPR scenario, and we hope any possible confusion is avoided from context. 
	
	In a Bell scenario, then, we want to characterise the set $\cQ_{\text{Bell}}$ of correlations $\pr_{\A\Bo|\X\Y}$ that admit a \textit{quantum realisation}. We recall here that $\pr_{\A\Bo|\X\Y}$ admits a quantum realisation\footnote{To be more precise, this is the definition of quantum realisation in the so-called \textit{tensor product paradigm}.} if there exists a Hilbert space $\cH_A$ associated to Alice, a Hilbert space $\cH_B$ associated to Bob, a bipartite quantum system in $\cH_A \otimes \cH_B$, a density matrix $\rho \in \mathcal{D}(\cH_A \otimes \cH_B)$, a collection of complete projective measurements $\{\{\Pi^{(\text{A})}_{a|x}\}_{a\in\A}\}_{x\in\X}$ in $\cH_A$, and a collection of complete projective measurements $\{\{\Pi^{(\text{B})}_{b|y}\}_{b\in\Bo}\}_{y\in\Y}$ in $\cH_B$, such that 
	\begin{align}
		p(ab|xy) = \tr{\Pi^{(\text{A})}_{a|x} \otimes \Pi^{(\text{B})}_{b|y} \, \rho}{} \,, \quad \forall \, a\in\A\,,\, b \in \Bo\,\, x \in \X\,,\, y \in \Y\,.
	\end{align}
	
	The question of characterising quantum correlations is then framed as: which principles do we need to impose on $\pr_{\A\Bo|\X\Y}$ so that all and only the ones compatible with them are the quantum ones $\cQ_{\text{Bell}}$? Notice then that such principles can then be formulated merely in terms of the probabilities themselves, by imposing constraints on the positive numbers in $\pr_{\A\Bo|\X\Y}$. Probabilities are objects with which we are familiar from probability theory, and the object itself is the same as we would study classically (apart from that, when going beyond classical, it is possible to general different correlations $\pr_{\A\Bo|\X\Y}$). All this is to say that there is a natural way to try and bound the correlations by formulating principles regarding only their statistical predictions\footnote{This approach to characterising correlations is sometimes referred to as \textit{device-independent}. Given how device-independent principles have failed to answer the question fully, it is now conjectured by some that device-independent principles are useful but not sufficient to characterise $\cQ_{\text{Bell}}$. }. Examples of these principles include \textit{no signalling} (that Alice and Bob cannot use $\pr_{\A\Bo|\X\Y}$ to communicate faster than the speed of light) \cite{popescu1994quantum}, \textit{non-trivial communication complexity} (that Alice cannot use $\pr_{\A\Bo|\X\Y}$ to send an unbounded amount of information to Bob by transmitting only a single bit of classical communication) \cite{van2013implausible}, and \textit{macroscopic locality} (when the source that prepares Alice and Bob's shared physical system sends them many independent copies of it rather than a single one, then the correlations observed in this macroscopic experiment\footnote{Under certain assumptions, see Ref.~\cite{ML}.} are classical even if the single-system correlations $\pr_{\A\Bo|\X\Y}$ are not) \cite{ML,MNC}, among others. 
	
	\bigskip
	
	Now let us get back to an EPR scenario, were the Alices have access to classical variables (denoting their measurement choices and outcomes) and Bob has access to a quantum system prepared in a known quantum state. Here, we not only want to characterise the outcome statistics of the Alices' measurements, but also Bob's Hilbert space and his system's quantum state. Hence, the tools and approach from the case of Bell scenarios cannot be directly applied to assemblages. The main challenge here is then what to incorporate in the formulation of the principles such that they capture Bob's quantum nature meaningfully. The types of axioms used to derive quantum theory are quite different from the type of device-independent principles pursued for characterising correlations in Bell scenarios, so the challenge is to find the crucial aspect of each approach that may be relevant for characterising assemblages in EPR scenarios. This comprehensive approach is beyond the scope of this work. Here we take the first steps into approaching the question by asking ourselves: (i) what information should we complement device-independent principles with so as to get a chance at bounding assemblages?, and (ii) how useful may such a first step be for tackling the question of interest? 
	
	\subsection{Bipartite EPR scenarios}
	
	\quad Traditional bipartite EPR scenarios are a very special case of EPR scenarios for various reasons, one of them being that, here, the question of `characterising quantum assemblages' has been fully answered:
	
	\begin{thm} \textbf{GHJW theorem \cite{gisin1989stochastic,hughston1993complete} .--}\label{GHJW} \\
		Consider a bipartite EPR scenario $(\A,\X,\cH_B)$. An assemblage $\Ass{\A}{\X}$ admits a quantum realisation if and only if it satisfies the No-Signalling principle. 
	\end{thm}
	Formally, the mathematical constraints that the No-Signalling principle implies are those of Eq.~\eqref{eq:NScons}. 
	There are various proofs of this theorem in the literature after the seminal work of Gisin \cite{gisin1989stochastic} and Hughston, Jozsa, and Wootters \cite{hughston1993complete}, using different proof techniques. In the following we include one such proof based on moment matrices, for illustration.
	
	\begin{proof}
		
		The structure of the proof is to show that: on the one hand, a non-signalling assemblage in this setting is an almost-quantum assemblage as per Def.~\ref{def:aqassemmm}; on the other hand, all almost-quantum assemblages are quantum-realisable assemblages in this setting.
		
		First, recall how almost-quantum assemblages are defined in terms of a moment matrix $\Gamma_{\mathrm{AQ}}^{\mathrm{EPR}}$ in the following way:
		\begin{compactitem}
			\item The set of words is $\mathcal{S}_{\mathrm{AQ}} := \{ \emptyset \} \cup \left\{ a|x  \Big\lvert  a \in \A \,,\, x\in\X \right\}$ as per Def.~\ref{def:AQw},
			\item $\Gamma_{\mathrm{AQ}}^{\mathrm{EPR}} (a|x,a'|x) = \mathbf{0}$ if $a\neq a'$, for all $a\neq a' \in \A$ and $x \in \X$ (Eq.~\eqref{eq:30}),
			\item $\Gamma_{\mathrm{AQ}}^{\mathrm{EPR}} (a|x,\emptyset) = \Gamma_{\mathrm{AQ}}^{\mathrm{EPR}} (\emptyset,a|x) = \Gamma_{\mathrm{AQ}}^{\mathrm{EPR}} (a|x,a|x)$ for all $a\in\A\,,\, x\in\X$ (Eq.~\eqref{eq:29}).
		\end{compactitem}
		
		For simplicity, in the rest of this proof we will write $\Gamma_{\mathrm{AQ}}^{\mathrm{EPR}}$ as $\Gamma$.
		
		Given a non-signalling assemblage there is a moment matrix $\Gamma$ such that $\Gamma (\emptyset,\emptyset):= \sigma_R$, $\Gamma (\emptyset,a|x) := \sigma_{a|x} =: \Gamma (a|x,\emptyset)$ for all $a\in\A\,,\, x\in\X$. To see that this $\Gamma$ is a valid moment matrix for the assemblage $\Ass{\A}{\X}$ notice that (i) conditions in Eqs.~\eqref{eq:29}, \eqref{eq:thetrace1}, \eqref{eq:30}, \eqref{eq:31}, and \eqref{eq:32} hold directly by construction, and (ii) condition in Eq.~\eqref{eq:28} can be shown by explicitly constructing a positive semi-definite matrix. Take $\sigma_R$, and write its purification as $\ket{\psi}=\sum_i\sqrt{\lambda_i}\ket{i}_{aux}\ket{i}_B$, with $\cH_{aux}$ being an auxiliary Hilbert space used for purification. Define now
		\begin{equation}
			N_{a|x}:=\frac{1}{\sqrt{\sigma_R}}\sigma_{a|x}\frac{1}{\sqrt{\sigma_R}};
		\end{equation}
		and notice that $N_{a|x}$ is a POVM. Thus there is another auxiliary system $\cH_{aux'}$ such that 
		\begin{equation}
			\tr{(N_{a|x}\otimes\mathbbm{1})\ket{\psi}\bra{\psi}}{aux}\equiv\tr{\Pi_{a|x}\ket{\Psi}\bra{\Psi}}{aux,aux'},
		\end{equation}
		where $\ket{\Psi}=\ket{\psi}\otimes\ket{0}_{aux'}$. Finally, define a matrix $V$ such that
		\begin{equation}
			V^\dagger(\emptyset):=\ket{\Psi};
		\end{equation}
		\begin{equation}
			V^\dagger(a|x):=\Pi_{a|x}\ket{\Psi}
		\end{equation}
		Here the terms in the brackets denote the particular words that will appear in the final moment matrix, as will become clear now. Define a positive semi-definite matrix $\Sigma=V^\dagger V$. From this matrix we can recover $\Gamma$ by tracing out the two auxiliary systems with spaces $\cH_{aux}$ and $\cH_{aux'}$, i.e.  $\Gamma=\tr{V^\dagger V}{aux,aux'}$. Since partial tracing is a CP map and $V^\dagger V\geq0$, thus $\Gamma\geq0$ and every non-signalling assemblage has an almost-quantum moment matrix.
		
		To prove that every almost-quantum moment matrix has a quantum realisation in this bipartite setting, from the moment matrix we can explicitly describe the state and measurements for Alice that give that moment matrix. Given a moment matrix $\Gamma$, the sub-matrix $\Gamma(\emptyset)$ is a Gramian matrix such that each element $\Gamma(\mathbf{0})_{ij}:=\bra{i}\Gamma(\emptyset)\ket{j}=\bra{u_j} u_i\rangle$ is the inner product of a set of (possibly sub-normalised) vectors $\{\ket{u_i}\in\mathcal{H}\}_i$ in some Hilbert space $\mathcal{H}_{A}$, where $A$ is used to denote Alice's system. From this set of vectors, we can construct a quantum state in the space $\mathcal{H}_{A}\otimes\mathcal{H}_{B}$:
		\begin{equation}
			\ket{\psi}=\sum_{i}\ket{u_{i}}\ket{i}_{B}
		\end{equation}
		where $\ket{i}_{B}$ is a vector in $\mathcal{H}_{B}$. It can readily be deduced that this gives the correct reduced density matrix, by noting:
		\begin{align}
			\tr{\ket{\psi}\bra{\psi}}{A}&=\sum_{k}\langle k|_{A}\left(\sum_i\ket{u_i}\ket{i}_{B}\sum_{j}\bra{u_j}\bra{j}_{B}\right)\ket{k}_{A}\\
			&=\sum_{ij}\bra{u_j}\sum_{k}\ket{k}\bra{k}_{A}\ket{u_i}\ket{i}\bra{j}_{B}\\
			&=\sum_{ij}\bra{u_j}u_i\rangle\ket{i}\bra{j}_{B}\\
			&=\Gamma(\emptyset)
		\end{align}
		It remains to describe Alice's measurements, which will be projectors $\{\Pi_{a|x}\}_{a,x}$ that project onto $\mathcal{H}_{A}$. Given a moment matrix $\Gamma$, from each sub-matrix $\Gamma_{a|x,a|x}$ we can describe a projective measurement. First note that, as before, each element satisfies $\bra{i}\Gamma(a|x,a|x)\ket{j}=\langle u_{j}^{a|x}\ket{u_{i}^{a|x}}$, thus with $a|x$ we can associate a set of vectors $\{\ket{u_{i}^{a|x}}\}_{i}$ living in the same Hilbert space $\cH_{A}$ as before; this is because both $\Gamma(\emptyset)$ and $\Gamma(a|x,a|x)$ have the same dimension.
		
		Now each of Alice's projectors $\Pi_{a|x}$ just projects onto the span of the set $\{\ket{u_{i}^{a|x}}\}_{i}$. First note that since $\Gamma_{\mathrm{AQ}}^{\mathrm{EPR}} (a|x,a'|x) = \mathbf{0}$ if $a\neq a'$, this implies that the span of the set  $\{\ket{u_{i}^{a'|x}}\}_{i}$ lies in the orthogonal complement of the span of the set $\{\ket{u_{i}^{a|x}}\}_{i}$. This property ensures that, for each $x$, $\Pi_{a|x}\Pi_{a'|x}=\delta_{a,a'}\Pi_{a|x}$. Because $\Gamma(\emptyset,a|x)=\Gamma(a|x,a|x)$, this implies that $\langle u_{j}\ket{u^{a|x}_i}=\langle u^{a|x}_j\ket{u^{a|x}_i}$, thus both $\ket{u_{j}}$ and $\ket{u^{a|x}_{j}}$ have the same inner product with all vectors $\ket{u^{a|x}_i}$ onto which $\Pi_{a|x}$ projects. Therefore, $\Pi_{a|x}\ket{u_i}=\ket{u^{a|x}_i}$ and that $\sum_{a}\Pi_{a|x}=\mathbbm{1}$ since $\sum_{a}\Gamma(\emptyset,a|x)=\Gamma(\emptyset)$. Finally, given that $\Pi_{a|x}\ket{u_i}=\ket{u^{a|x}_i}$ one can then infer that this complete, projective measurement $\{\Pi_{a|x}\}_{a}$ recovers the assemblage element $\sigma_{a|x}$ in the following way:
		\begin{align}
			\sigma_{a|x}&=\tr{\Pi_{a|x}\otimes\mathbbm{1}\ket{\psi}\bra{\psi}}{A}\\
			&=\sum_{k}\langle k|_{A}\left(\sum_i\Pi_{a|x}\ket{u_i}\ket{i}_{B}\sum_{j}\bra{u_j}\bra{j}_{B}\right)\ket{k}_{A}\\
			&=\sum_{ij}\bra{u_j}\sum_{k}\ket{k}\bra{k}_{A}\ket{u^{a|x}_i}\ket{i}\bra{j}_{B}\\
			&=\sum_{ij}\bra{u_j}u_i^{a|x}\rangle\ket{i}\bra{j}_{B}\\
			&=\Gamma(\emptyset,a|x),
		\end{align}
		thus completing the proof.
	\end{proof}
	
	From this we see that characterising quantum-realisable assemblages is a fundamentally \textit{multipartite} task. Another way of viewing this is that the non-signalling principle is enough to capture bipartite quantum assemblages.

	\section{Characterising assemblages by characterising correlations}
	
	\quad In this section we discuss how device-independent principles defined to bound correlations in Bell scenarios can be re-purposed to bound the set of assemblages in EPR scenarios. Since there is a close connection between correlations and assemblages, the hope is to leverage our understanding of the former to the case of assemblages.
	
	The first step we take is motivated by the practicalities of the EPR experiment: at the end of the day, for Bob to characterise any assemblage he needs to perform tomography on the state of his system. That is, he will have access to a set of tomographically-complete measurements $\text{TC}_\text{B} := \{\{M_{b|y}\}_b\}_y$ in $\cH_B$, and all the information needed to reconstruct the elements of $\Ass{\vec{\A}}{\vec{\X}}$ is given by 
	\begin{align}\label{eq:tccorr}
		\pr_{\vec{\A}\Bo|\vec{\X}\Y}^{\text{\textbf{T}}} := \{p(\vec{a}b|\vec{x}y) \equiv \tr{M_{b|y}\sigma_{a_1...a_\Np|x_1...x_\Np}}{} \ \vert \ M_{b|y}\in \text{TC}_\text{B} \ , \ \sigma_{a_1\ldots a_\Np |x_1 \ldots x_\Np} \in \Ass{\vec{\A}}{\vec{\X}} \}\,.
	\end{align}
	For each $y$, the elements $\{M_{b|y}\}_b$ can be taken to be projectors without loss of generality. 
	
	Given a particular but arbitrary principle for correlations in Bell scenarios, one may then apply it to assemblages as follows: 
	
	\begin{defn}\label{defn}
		Let $\cP$ be a principle for correlations in Bell scenarios, We say that an assemblage $\Ass{\vec{\A}}{\vec{\X}}$ satisfies $\cP$ if the correlation $\pr_{\vec{\A}\Bo|\vec{\X}\Y}^{\text{\textbf{T}}}$ from Eq.~\eqref{eq:tccorr} -- produced when Bob performs tomographically-complete measurements on the elements of $\Ass{\vec{\A}}{\vec{\X}}$ -- satisfies $\cP$.
	\end{defn}

	\subsection{Macroscopic classicality considerations}\label{se:4.1}
	
	\quad We first illustrate how to use our idea to impose on assemblages the Macroscopic Locality principle \cite{ML}, or more precisely, its stronger version called Macroscopic Non-contextuality \cite{MNC}. 
	
	The principle of Macroscopic Non-contextuality states that a correlation $\pr_{\vec{\A}\Bo|\vec{\X}\Y}^{\text{\textbf{T}}} $ is compatible with a particular classical limit\footnote{The particular aspects of this classical limit go beyond the scope of this manuscript, and are discussed elsewhere \cite{vini2}.} if and only-if  $\pr_{\vec{\A}\Bo|\vec{\X}\Y}^{\text{\textbf{T}}} $ belongs to the so-called almost-quantum set of correlations. So the question is: what is the set of assemblages whose correlations $\pr_{\vec{\A}\Bo|\vec{\X}\Y}^{\text{\textbf{T}}} $ are all and only almost-quantum correlations? We discuss below answers to this question through Propositions \ref{lemma} and \ref{conj}.
	
	\begin{prop}\label{lemma}
		Almost-quantum assemblages satisfy macroscopic non-contextuality. 
	\end{prop}
	
	\begin{proof} \quad \\
		An assemblage $\Ass{\vec{\A}}{\vec{\X}}$ is almost-quantum assemblage if there exist \cite[Lemma 16]{hoban2018channel}:
		\begin{compactitem}
			\item A Hilbert space $\cH:=\mathcal{K}\otimes\cH_B$;
			\item A state $\ket{\psi}\in\cH$;
			\item A collection of projective measurements $\{\{\Pi^{(i)}_{a_i|x_i}\}_{a_i\in\A}\}_{x_i\in\X}\}$ \blk acting on $\mathcal{K}$, for each party ${i=1,\dots,\Np}$,
		\end{compactitem}
		such that 
		\begin{equation}
			\sigma_{a_1...a_\Np|x_1...x_\Np}:=\tr{\Pi^{(1)}_{a_1|x_1}...\Pi^{(\Np)}_{a_\Np|x_\Np}\otimes\mathbbm{1}_B\ket{\psi}\bra{\psi}}{\mathcal{K}},
		\end{equation}
		and
		\begin{equation}
			\Pi^{(1)}_{a_1|x_1}\ldots \, \Pi^{(\Np)}_{a_\Np|x_\Np}\otimes\mathbbm{1}_B\ket{\psi}=\Pi^{(\pi(1))}_{a_{\pi(1)}|x_{\pi(1)}}\ldots \, \Pi^{(\pi(\Np))}_{a_{\pi(\Np)}|x_{\pi(\Np)}}\otimes\mathbbm{1}_B\ket{\psi},
		\end{equation}
		for all permutations $\pi$ of the $\Np$ Alices.
		
		Let $\{\Pi_{b|y}\}_{b\in\Bo\,,\,y\in\Y}$ be a collection of projective measurements for Bob (not necessarily a tomographically complete one), that he could perform over the elements of $\Sigma_{\vec{\A}|\vec{\X}}$.

		The correlations $\pr_{\vec{\A}\Bo|\vec{\X}\Y}$ arising from these measurements then satisfy the following property: there exist
		\begin{compactitem}
			\item a Hilbert space $\cH$;
			\item a state $\ket{\psi}\in\cH$;
			\item a collection of projective measurements $\{\Pi^{(i)}_{a_i|x_i}\}_{a_i\in\A,x_i\in\X,i=1,\ldots,\Np}\cup\{\Pi_{b|y}\}_{b\in\Bo,y\in\Y}$ acting on $\cH$,
		\end{compactitem}
		such that
		\begin{equation}\label{eq:19}
			p(\vec{a}b|\vec{x}y):=\trace{\Pi_{b|y}\sigma_{a_1...a_\Np}}=\trace{\Pi^{(1)}_{a_1|x_1}...\Pi^{(\Np)}_{a_\Np|x_\Np}\otimes\Pi_{b|y}\ket{\psi}\bra{\psi}}. 
		\end{equation}
		
		Now, since Bob's measurements commute with all of Alices' measurements, then Eq.~\eqref{eq:19} provides an almost-quantum realisation of the correlations $\pr_{\vec{\A}\Bo|\vec{\X}\Y}$.
		
		The last step is to recall that almost-quantum correlations satisfy the macroscopic noncontextuality principle \cite{MNC}. Hence, almost-quantum assemblages also satisfy the principle in the sense of Def.~\ref{defn}.
	\end{proof}
	
	\bigskip
	
	We have therefore shown that almost-quantum assemblages satisfy the macroscopic non-contextuality principle. Now, is the converse true? That is, is this an `if and only if' result? The answer to this is not straightforward, and requires one to leverage some physical characteristics of an EPR experiment, as we discuss next.
	
	For the argument, notice that there exist assemblages $\Ass{\vec{\A}}{\vec{\X}}$ such that:
	\begin{compactitem}
		\item the correlations $\pr_{\vec{\A}\Bo|\vec{\X}\Y} $ that may arise when Bob performs any measurements on his system (including but not restricted to \ tomographically-complete ones) admit a quantum realisation,
		\item the assemblage $\Ass{\vec{\A}}{\vec{\X}}$ is not almost-quantum. 
	\end{compactitem}
	As we will discuss in detail later, in \cite{sainz2015postquantum}, for instance, such an example of an assemblage was given. One way to understand this it that the quantum realisation of $\pr_{\vec{\A}\Bo|\vec{\X}\Y} $ may require a quantum system in Bob's wing that has larger dimension than that of $\cH_B$ (Bob's Hilbert space in the EPR scenario), hence explaining why the correlations admit a quantum realisation but the assemblage does not.
	From this fact it follows that
	there exist assemblages that generate almost-quantum correlations but are themselves not almost-quantum.

	The examples of such assemblages might come across as evidence that assemblages beyond the almost-quantum set satisfy the macroscopic non-contextuality principle. And this indeed highlights a deficiency of Def.~\ref{defn}: one needs additional information beyond the correlations $\pr_{\vec{\A}\Bo|\vec{\X}\Y}^{\text{\textbf{T}}} $ to fully flesh-out the structure of assemblages in EPR scenarios. See Sec.~\ref{sse:example} for a more comprehensive device-dependent discussion of the claim stated here.
	
	\subsection{The need for device-dependent principles}
	
	\quad The task that we set out to explore is that of characterising quantum assemblages, that is, a collection of quantum states. Hence, it comes as no surprise that one may require principles that go beyond the device-independent ones. In the last section, however, we saw how one particular device-dependent way in which device-independent principles may be applied to assemblages failed at answering our question. Here we briefly discuss other device-dependent considerations one may leverage.
	
	One of the captivating properties of post-quantum EPR inference is that it is a new post-quantum phenomenon in its own right, which is not merely implied by post-quantumness in Bell scenarios. However, when it comes to characterising assemblages solely from device-independent principles this becomes a roadblock. More precisely, consider an assemblage $\Ass{\vec{\A}}{\vec{\X}}$, and a set of measurements $\{\{N_{b|y}\}_{b\in \Bo}\}_{y\in\Y}\}$ on $\cH_B$ for Bob (not necessarily tomographically complete), where $\Y$ is the set of classical labels for the measurement choices, and $\Bo$ the set of classical labels for the measurements' outcomes. One may compute the correlations
	\begin{align}
		\pr_{\vec{\A}\Bo|\vec{\X}\Y} := \{p(\vec{a}b|\vec{x}y) \equiv \tr{N_{b|y}\sigma_{a_1...a_\Np|x_1...x_\Np}}{} \ \vert \ y\in\Y \ ,\ b \in \Bo \ , \ \sigma_{a_1...a_\Np|x_1...x_\Np} \in \Ass{\vec{\A}}{\vec{\X}} \}\,.
	\end{align}
	Notice that the operational constraints in the scenario where the Alices and Bob generate the correlations $\pr_{\vec{\A}\Bo|\vec{\X}\Y}$ are consistent with those of a Bell scenario, and hence we can think of $\pr_{\vec{\A}\Bo|\vec{\X}\Y}$ as being `correlations in a Bell experiment'. Now, Refs.~\cite{sainz2015postquantum,sainz2018formalism} show that there exist post-quantum assemblages $\Ass{\vec{\A}}{\vec{\X}}$ that can only generate quantum correlations $\pr_{\vec{\A}\Bo|\vec{\X}\Y}$ in such a Bell scenario. Therefore, any principle imposed on $\Ass{\vec{\A}}{\vec{\X}}$ which only looks at the possible correlations $\pr_{\vec{\A}\Bo|\vec{\X}\Y}$ that it may generate \textit{without involving any information on how these correlations were generated}, will immediately be satisfied by some post-quantum assemblages, including assemblages with no almost-quantum realisation. 
	
	In the previous subsection we discussed a particular way in which device-dependent considerations can be brought into the picture to bound EPR assemblages. This approach, presented in Def.~\ref{defn}, complements these device-independent principles with the particular information on the dimension of Bob's system $\cH_B$ as well as the particular characterisation of his measurement device which is  also required to perform tomographically-complete measurements. What Def.~\ref{defn} does not explicitly state -- and which some may feel inclined to include it in retrospect -- is that when thinking of $\pr_{\vec{\A}\Bo|\vec{\X}\Y}^{\text{\textbf{T}}}$ as a `correlation in a Bell scenario' one can moreover require that the realisation of $\pr_{\vec{\A}\Bo|\vec{\X}\Y}^{\text{\textbf{T}}}$ -- be it quantum, almost-quantum, etc -- should have Bob's Hilbert space factorised from whatever object represents Alices' systems, and his quantum actions should act locally on $\cH_B$. In the next subsection we show how to apply this reasoning when asking how the macroscopic non-contextuality principle constrains EPR assemblages. 
	
	\subsection{A stronger device-dependent version of macroscopic non-contextuality}\label{sse:example}
	
	\quad We saw in Section \ref{se:4.1} how the minimalistic application of the macroscopic non-contextuality principle to assemblages via Def.~\ref{defn} fails at constraining the set of assemblages in a satisfactory way. The question we discuss here is how we can complement Def.~\ref{defn} with device-dependent requirements so that a more meaningful set of assemblages is singled out by the macroscopic non-contextuality principle. 
	
	The first consideration that stems from the EPR scenario is that Bob's Hilbert space is known and specified ($\cH_B$), hence one can argue that the required specific (such as quantum or almost-quantum) realisation of $\pr_{\vec{\A}\Bo|\vec{\X}\Y}^{\text{\textbf{T}}} $ does happen in a factorised Hilbert space $\cH:=\mathcal{K}\otimes\cH_B$. In the same way, a natural requirement is that Bob's measurement operators that realise the correlations only act on $\cH_B$. 
	
	A second consideration is about the Alices. Given the previous discussion, if one requires the correlations $\pr_{\vec{\A}\Bo|\vec{\X}\Y}^{\text{\textbf{T}}} $ to be realised in $\cH:=\mathcal{K}\otimes\cH_B$ one may also take a step forward demanding that the operations implemented by the Alices when realising the correlations only act on $\mathcal{K}$. While this is a more questionable assumption, it is worth considering, although it is not required for the proof of our following claim.
	
	\begin{prop}\label{conj}
		Consider an EPR scenario $(\vec{\X},\vec{\A},\cH_B)$, and an assemblage $\Ass{\vec{\A}}{\vec{\X}}$ representing the knowledge on the state of Bob's system. Let $\pr_{\vec{\A}\Bo|\vec{\X}\Y}^{\text{\textbf{T}}} $ be the correlation associated to $\Ass{\vec{\A}}{\vec{\X}}$ when Bob performs tomographically-complete measurements on the latter. Assume that:
		\begin{compactitem}
			\item $\pr_{\vec{\A}\Bo|\vec{\X}\Y}^{\text{\textbf{T}}} $ is an almost-quantum correlation, that is, that $\Ass{\vec{\A}}{\vec{\X}}$ satisfies the macroscopic non-contextuality principle as per Def.~\ref{defn},
			\item the almost-quantum realisation of $\pr_{\vec{\A}\Bo|\vec{\X}\Y}^{\text{\textbf{T}}} $ is achieved in a factorised Hilbert space $\cH:=\mathcal{K}\otimes\cH_B$ by measurements $\{\Pi_{b|y}\}_{b\in\mathbb{B},y\in\Y}$ for Bob with the form $\Pi_{b|y}=\mathbbm{1}_{\mathcal{K}}\otimes\Pi^{TC}_{b|y}$, $\forall b\in\mathbb{B},y\in\Y$.
		\end{compactitem}
		Then, $\Ass{\vec{\A}}{\vec{\X}}$ is an almost-quantum assemblage. 
	\end{prop}

	\begin{proof}
		Consider the almost quantum certificate $\Gamma_{\mathrm{AQ}}^{\mathrm{B}}$ for the correlations $\pr_{\vec{\A}\Bo|\vec{\X}\Y}^{\text{\textbf{T}}}$. Denoting by $\mathcal{S}_{\mathrm{AQ}}$ the almost quantum set of words  for all black-box parties' labels, it is the case that $\Gamma_{\mathrm{AQ}}^{\mathrm{B}}\in M_{|\mathcal{S}_{\mathrm{AQ}}|\cdot|\Bo|\cdot|\Y|}(\mathbb{R})$, where $M_n(\mathbb{R})$ is the space of $n\times n$ matrices with real elements. Furthermore, there is an isomorphism between this space and $M_{|\mathcal{S}_{\mathrm{AQ}}|}(M_{|\Bo|\cdot|\Y|}(\mathbb{R}))$, i.e., the space of $|\mathcal{S}_{\mathrm{AQ}}|\times|\mathcal{S}_{\mathrm{AQ}}|$ matrices with elements from $M_{|\Bo|\cdot|\Y|}(\mathbb{R})$. From this, we can represent $\Gamma_{\mathrm{AQ}}^{\mathrm{B}}$ as
		\begin{equation}
			{\Gamma}_{\mathrm{AQ}}^{\mathrm{B}}=\sum_{\mathbf{v},\mathbf{w}\in\mathcal{S}_{\mathrm{AQ}}}\ket{\mathbf{v}}\bra{\mathbf{w}}\otimes\Lambda_{\mathbf{v},\mathbf{w}}\,,
		\end{equation}
		where each $\Lambda_{\mathbf{v},\mathbf{w}}\in M_{|\Bo|\cdot|\Y|}(\mathbb{R})$ is a submatrix of $\Gamma_{\mathrm{AQ}}^{\mathrm{B}}$. Notice, moreover, that when $\mathbf{w}^\dagger\mathbf{v} \in \mathcal{S}_{\mathrm{AQ}}$, $\Lambda_{\mathbf{v},\mathbf{w}}$ contains the probabilities and correlations of Bob's tomographic measurements all conditioned to the labels $\mathbf{w}^\dagger\mathbf{v}$ of the black-box parties.
		
		Because we know that Bob's measurements act only on $\cH_B$ in the factorized Hilbert space, and because they are tomographically complete, there is an isomorphism between the space $M_{|\Bo|\cdot|\Y|}(\mathbb{R})$ and Bob's Hilbert space $\cH_B$\footnote{If the Hilbert space were not factorized, or if Bob's measurements wouldn't act on $\cH_B$ only, the tomographic reconstruction could yield non-physical states. In other words, the isomorphism is in fact between a subset of $M_{|\Bo|\cdot|\Y|}(\mathbb{R})$, whose entries are correlations arising from Bob's measurements acting on $\cH_B$, and $\cH_B$ itself.}, i.e., there is a bijective, positive map $\mathcal{T}(\cdot)$ such that:
		\begin{align}
			\mathcal{T}(\Lambda_{\emptyset,\emptyset})=\sigma_R\,;\\
			\mathcal{T}(\Lambda_{\emptyset,\mathbf{v}})=\mathcal{T}(\Lambda_{\mathbf{v},\mathbf{v}}) =\sigma_{\mathbf{v}}\,.
		\end{align}
		These elements ($\sigma_R$ and $\sigma_{\mathbf{v}}$) are positive semidefinite by definition, since the corresponding $\Lambda_{\emptyset,\emptyset}$ and $\Lambda_{\mathbf{v},\mathbf{v}}$ are positive semidfinite themselves (the latter follows from Sylverster's criterion, since $\Lambda_{\emptyset,\emptyset}$ and $\Lambda_{\mathbf{v},\mathbf{v}}$ are  principal submatrices of $\Gamma_{\mathrm{AQ}}^{\mathrm{B}}$ and hence are also positive semidefinite).\\ 
		
		From the linearity of $\mathcal{T}(\cdot)$, we can define
		\begin{equation}
			\Gamma_{\mathrm{AQ}}^{\mathrm{EPR}}:=(\mathbbm{1}\otimes\mathcal{T})[\tilde{\Gamma}_{\mathrm{AQ}}^{\mathrm{B}}]=\sum_{\mathbf{v},\mathbf{w}\in\mathcal{S}_{\mathrm{AQ}}}\ket{\mathbf{v}}\bra{\mathbf{w}}\otimes\mathcal{T}(\Lambda_{\mathbf{v},\mathbf{w}})\,.
		\end{equation}
		This matrix has labels and columns completely specified by the words in $\mathcal{S}_{\mathrm{AQ}}$, each entry being a matrix in $\cH_B$. It also has the following properties:
		\begin{align}
			&\Gamma_{\mathrm{AQ}}^{\mathrm{EPR}} (\mathbf{v},\mathbf{w}) = \Gamma_{\mathrm{AQ}}^{\mathrm{EPR}} (\mathbf{v'},\mathbf{w'}) \quad \text{if} \quad \mathbf{v}^\dagger\mathbf{w} = \mathbf{v'}^\dagger\mathbf{w'} \,,\\ 
			&\tr{\Gamma_{\mathrm{AQ}}^{\mathrm{EPR}} (\emptyset,\emptyset)}{} = 1 \,, \\ 
			&\Gamma_{\mathrm{AQ}}^{\mathrm{EPR}} (\mathbf{v},\mathbf{w}) = \mathbf{0} \quad \text{if} \quad \mathbf{v}^\dagger\mathbf{w} \quad \text{is null} \,,
		\end{align}
		where $\mathbf{0}$ is a $d\times d$ matrix with null elements, $d$ being the dimension of $\cH_B$. The first property follows from the fact that $\Lambda_{\mathbf{v},\mathbf{w}}=\Lambda_{\mathbf{v}',\mathbf{w}'}$ when $\mathbf{v}^\dagger\mathbf{w}=\mathbf{v'}^\dagger\mathbf{w'}$ and this equality is preserved by the isomorphism. The second comes from the constraint that $\Lambda_{\emptyset,\emptyset}$ has as entries $\trace{(\mathbbm{1}_{\mathcal{K}}\otimes\Pi_{b|y}^{TC})\ket{\psi}\bra{\psi}}$, the tomographic data that is mapped into Bob's partial state $\tr{\ket{\psi}\bra{\psi}}{\mathcal{K}}$. The third arises from $\Lambda_{\mathbf{v},\mathbf{w}}=\mathbf{0}$ whenever $\mathbf{v}^\dagger\mathbf{w}$ is null and the linearity of the isomorphism.\\
		
		To finally show that $\Gamma_{\mathrm{AQ}}^{\mathrm{EPR}}\geq0$, notice that there is a state $\ket{\psi}\in\mathcal{K}\otimes\cH_B$ and projectors $\{\Pi^{(i)}_{a_i|x_i}\}_{a\in\A,x\in\X,i=1,...,\Np}$ acting on this Hilbert space, such that
		\begin{equation}
			\mathcal{T}(\Lambda_{\mathbf{v},\mathbf{w}})=\tr{\left(\prod_{(a_i|x_i)\in\mathbf{v}}\Pi_{a_i|x_i}^{(i)}\right)^\dagger\prod_{(a'_j|x'_j)\in\mathbf{w}}\Pi_{a'_j|x'_j}^{(j)}\ket{\psi}\bra{\psi}}{\mathcal{K}}.
		\end{equation}
		The entry $\Gamma_{\mathrm{AQ}}^{\mathrm{EPR}}(\mathbf{v}_l,\mathbf{w}_m)$, where $l,m$ are labels of the rows and columns in $\cH_B$, can be written as
		\begin{equation}
			\Gamma_{\mathrm{AQ}}^{\mathrm{EPR}}(\mathbf{v}_l,\mathbf{w}_m)=\sum_{\ket{k}\in\mathcal{K}}\bra{k}\braket{l|\prod_{(a'_j|x'_j)\in\mathbf{w}}\Pi_{a'_j|x'_j}^{(j)}\ket{\psi}\bra{\psi}\left(\prod_{(a_i|x_i)\in\mathbf{v}}\Pi_{a_i|x_i}^{(i)}\right)^\dagger|k}\ket{m}\,,
		\end{equation}
		or reorganizing and including an identity $\sum_{\ket{k'}\in\mathcal{K}}\ket{k'}\bra{k'}$,
		\begin{equation}
			\Gamma_{\mathrm{AQ}}^{\mathrm{EPR}}(\mathbf{v}_l,\mathbf{w}_m)=\sum_{\ket{k}\in\mathcal{K}}\braket{\psi|\left(\prod_{(a_i|x_i)\in\mathbf{v}}\Pi_{a_i|x_i}^{(i)}\right)^\dagger|k}\ket{m}\bra{k}\cdot\sum_{\ket{k'}\in\mathcal{K}}\bra{k'}\braket{l|\prod_{(a'_j|x'_j)\in\mathbf{w}}\Pi_{a'_j|x'_j}^{(j)}|\psi}\ket{k'},
		\end{equation}
		which can be rewritten as
		\begin{equation}
			\Gamma_{\mathrm{AQ}}^{\mathrm{EPR}}(\mathbf{v}_l,\mathbf{w}_m)=\braket{\mathbf{v}_l|\mathbf{w}_m}\,.
		\end{equation}
		It is thus possible to define a matrix $V$ whose columns are the vectors $\ket{\mathbf{v}_l}$, $\forall \mathbf{v}\in\mathcal{S}_{\mathrm{AQ}}$ and $l=1,...,\mbox{dim}\cH_B$, such that
		\begin{equation}
			\Gamma_{\mathrm{AQ}}^{\mathrm{EPR}}=V^\dagger V,
		\end{equation}
		and therefore, positive semidefinite.\\
	\end{proof}
	
	\section{Discussion}
	
	\quad In this work we considered the question of how to characterise a natural set of assemblages from basic physical principles in EPR scenarios -- a question that has been vastly explored in the device-independent framework for correlations in Bell scenarios. The hope is that if a wise set of principles is identified, the natural set of assemblages that emerges will be that of the quantum-realisable ones, gaining a valuable physical intuition behind quantum assemblages. For the case of correlations in Bell scenarios this question still remains open, and there is no reason a priori to expect that it becomes less difficult when asked for assemblages in EPR scenarios.
	
	In this work we explored a particular avenue towards characterising EPR assemblages: the idea is to map an assemblage $\Ass{\vec{\A}}{\vec{\X}}$ into a conditional probability distribution $\pr_{\vec{\A}\Bo|\vec{\X}\Y}^{\text{\textbf{T}}}$ in a particular Bell scenario by letting the characterised party (Bob) perform a set of tomographically-complete measurements on their quantum system. We hence explore how the assemblages $\Ass{\vec{\A}}{\vec{\X}}$ are constrained by imposing device-independent principles on their corresponding correlations $\pr_{\vec{\A}\Bo|\vec{\X}\Y}^{\text{\textbf{T}}}$. The motivation is that, since Bob performs tomographically complete measurements, all the information about the assemblage is encoded in $\pr_{\vec{\A}\Bo|\vec{\X}\Y}^{\text{\textbf{T}}}$, hence one would expect that exploring the characterisations of the latter would be sufficient. However, we find that the tensor-product structure of the joint Hilbert space underpinning the EPR experiment -- in particular, that Bob's Hilbert space is tensored to that of the Alices -- is not a property one may capture by merely studying the correlations $\pr_{\vec{\A}\Bo|\vec{\X}\Y}^{\text{\textbf{T}}}$, and hence the need for truly device-dependent principles for characterising EPR assemblages emerges.
	
	To explore the extent to which our particular avenue may characterise EPR assemblages from their correlations, we focused on the case of correlations $\pr_{\vec{\A}\Bo|\vec{\X}\Y}^{\text{\textbf{T}}}$ that satisfy the macroscopic non-contextuality principle. Here we showed that almost-quantum assemblages satisfy the macroscopic non-contextuality principle -- in other words, the correlations $\pr_{\vec{\A}\Bo|\vec{\X}\Y}^{\text{\textbf{T}}}$ that arise from almost-quantum assemblages satisfy macroscopic non-contextuality. The converse statement -- that assemblages which satisfy macroscopic non-contextuality are almost-quantum -- needs to be formalised in more detail since, as we argued before, the tensor-product structure of the Hilbert space into the product of the parties' needs to be assumed. Grounded in this additional assumption, we further showed that an assemblage satisfies the macroscopic non-contextuality principle only if it is an almost-quantum assemblage.
	
	Our proof techniques involve EPR inferential tools that are still novel, and we hope our work highlights their usefulness for future research in the field. Looking ahead, we have taken the first steps toward tackling the characterisation of EPR inference, and this has open -- as is usually the case in foundational research -- the door to a plethora of questions, a couple of which are the following. 
	First, one may wonder whether stronger constraints on EPR assemblages can be found by relating them to the objects of study in other non-classicality experiments -- for example, correlations in contextuality experiments or network non-classicality rather than the Bell scenarios considered here. Alternatively, one can explore the more challenging avenue of characterising assemblages via device-dependent principles formulated directly for EPR scenarios, without having better-known phenomena (such as Bell non-classicality or contextuality) mediate between them. In this case, toy theories that include quantum systems -- such as Witworld \cite{cavalcanti2021witworld} -- may serve as a test-bed for the strength of the proposed principles. 
	
	\section*{Acknowledgments}
	
	\quad We thank Valerio Scarani for fruitful discussions during the Vienna Quantum Foundations Conference 2020 (hosted in September 2021). 
	VPR and ABS acknowledge support by the Foundation for Polish Science (IRAP project, ICTQT, contract no. 2018/MAB/5, co-financed by EU within Smart Growth Operational Programme). 
	MJH and ABS acknowledge the FQXi large grant ``The Emergence of Agents from Causal Order'' (FQXi FFF Grant number FQXi-RFP-1803B). 
	The diagrams within this manuscript were prepared using TikZit.
	
	\bibliography{MUJPA}
	\bibliographystyle{unsrt}
	
\end{document}